\documentclass[11pt]{article}
\usepackage{amsmath}
\usepackage{amsthm}
\usepackage{titlesec}
\usepackage{indentfirst}

\newtheorem{theorem}{Theorem}

\newtheorem{lemma}[theorem]{Lemma}

\newtheorem{remark}[theorem]{Remark}

\begin{document}
\title{ SOME LANGUAGES RECOGNIZED BY TWO-WAY FINITE AUTOMATA WITH QUANTUM AND CLASSICAL STATES \footnote{This work is supported in
part by the National Natural Science Foundation (Nos.
60873055, 61073054, 61100001), the
Natural Science Foundation of Guangdong Province of China (No. 10251027501000004), the Fundamental Research Funds for the Central Universities (Nos. 10lgzd12,11lgpy36), the Research Foundation for the Doctoral Program of Higher School of Ministry
of Education (Nos. 20100171110042, 20100171120051) of China, the China Postdoctoral Science Foundation project (Nos. 20090460808, 201003375),
and the project of  SQIG at IT, funded by FCT and EU FEDER projects
projects QSec PTDC/EIA/67661/2006, AMDSC UTAustin/MAT/0057/2008,  NoE
Euro-NF, and IT Project QuantTel.} }

\author{Shenggen Zheng$^{1,}$\thanks{{\it  E-mail
address:} zhengshenggen@gmail.com},\hskip
2mm Daowen Qiu$^{1,2,3,}$\thanks{Corresponding author. {\it E-mail address:}
issqdw@mail.sysu.edu.cn (D. Qiu)},
\hskip 2mm Lvzhou Li$^{1,}$
\thanks{{\it  E-mail address:} lilvzhou@gmail.com}
 \\
\small{{\it $^{1}$Department of
Computer Science, Sun Yat-sen University, Guangzhou 510006,
  China }}\\
\small {{\it $^{2}$ SQIG--Instituto de Telecomunica\c{c}\~{o}es, Departamento de Matem\'{a}tica,}}\\
\small {{\it  Instituto Superior T\'{e}cnico, TULisbon, Av. Rovisco Pais
1049-001, Lisbon, Portugal}}\\
\small{{\it $^{3}$ The State Key Laboratory of Computer Science, Institute of Software,}}\\
\small{{ \it Chinese  Academy of Sciences, Beijing 100080, China}}
}

\date{ }
\maketitle \vskip 2mm \noindent {\bf Abstract}
\par
{\it Two-way finite automata with quantum and classical states} (2QCFA) were introduced by Ambainis and Watrous, and it was shown that 2QCFA have superiority over {\it two-way probabilistic finite automata} (2PFA) for recognizing some non-regular languages  such as the language $L_{eq}=\{a^{n}b^{n}\mid n\in \mathbf{N}\}$ and the palindrome language $L_{pal}=\{\omega\in \{a,b\}^*\mid\omega=\omega^R\}$, where $x^R$ is $x$ in the reverse order. It is interesting to find more languages like these that witness the superiority of 2QCFA over 2PFA. In this paper, we consider the language  $L_{m}=\{xcy\mid  \Sigma=\{a, b, c\}, x,y\in\{a,b\}^{*},c\in\Sigma, |x|=|y|\}$ that is similar to the middle language $L_{middle}=\{xay\mid x,y\in\Sigma^{*},a\in\Sigma, |x|=|y|\}$. We prove that
the language  $L_{m}$ can be recognized by 2QCFA with one-sided error in polynomial expected time. Also,  we show that $L_{m}$ can be recognized by 2PFA with bounded error, but only  in exponential expected time. Thus $L_{m}$ is another witness of the fact that 2QCFA are more powerful than their classical counterparts.

\par
\vskip 2mm {\sl Keywords:}   Computing models; Probabilistic finite automata; Quantum
finite automata.
\vskip 2mm

\section{Introduction}
 Interest in quantum computation has steadily increased since Shor's quantum algorithm for factoring integers in polynomial time \cite{Shor} and Grover's algorithm of searching in database of size $n$ with only $O(\sqrt{n})$ accesses \cite{Grover}. Clarifying the power of some fundamental models of quantum computation has attracted wide attentions in the academic community \cite{Gru99,Nie}. As we know, algorithms based on {\it quantum Turing machines} are complicated to implement using today's experiment technology. Therefore, it is natural to consider much more restricted quantum computing models.

As one of the simplest computing models, {\it
deterministic finite automata} (DFA) and {\it probabilistic finite automata (PFA)} have been deeply studied \cite{Hop,AP}.   Correspondingly, it may be interesting to consider restricted {\it
quantum Turing machines}, such as {\it quantum finite automata}
(QFA). QFA can be thought of as a theoretical model of quantum
computers in which the memory is finite and described by a finite-dimensional state space \cite{Amb06}, as {\it finite automata} (FA) are a natural model for classical computing with finite memory \cite{Hop}. QFA were first introduced independently by Kondacs and
Watrous \cite{Kon97}, as well as Moore and Crutchfield
 \cite{Moore}. As a quantum variant of FA, QFA have attracted wide
attentions in the academic community
 \cite{Amb06,LiQ2,LiQ3,Qiu4,AY,AY3}. There are many kinds of QFA having been proposed and studied (e.g., see \cite{LiQ4}). The study of QFA is mainly divided into two kinds:
one is {\it one-way quantum finite automata} (1QFA) whose tape heads move one square
right at each evolution, and the other {\it two-way quantum finite automata} (2QFA), in which
the tape heads are allowed to move towards right or left, or to be stationary.

Furthermore, according
to the measurement times in a computation, 1QFA have two fashions: measure-once 1QFA
(MO-1QFA) proposed by Moore and Crutchfield \cite{Moore} and measure-many 1QFA (MM-1QFA)
studied first by Kondacs and Watrous \cite{Kon97}.  MO-1QFA are strictly less powerful than their classical counterparts DFA \cite{Kon97,Moore}, since they  recognize only a proper subset of {\it regular languages} (RL).  Though MM-1QFA are more powerful than MO-1QFA \cite{Amb-F}, they still recognize  with bounded error  only a  proper subset of RL \cite{Bro}.

2QFA, however, are more powerful than their
classical counterparts.  2QFA can not only recognize all regular languages, but also recognize the non-regular language
$L_{eq}=\{a^{n}b^{n}\mid n\in \mathbf{N}\}$ with bounded error in linear time. Note that {\it two-way deterministic finite automata} (2DFA) recognize  the same family of languages  as DFA \cite{Hop}, and a 2PFA requires
exponential expected time to recognize  $L_{eq}$ \cite{RF,AA}. Therefore, 2QFA are more powerful than their
classical counterparts. However, 2QFA have  a disadvantage in
the sense that we need at least $\mathbf{O}(\log n)$ qubits to
store the positions of the tape head, which is relative to the
length of the input.

In order to conquer the above disadvantage, Ambainis and Watrous \cite{AJ} proposed a different two-way quantum
computing model $\raisebox{0.5mm}{---}$
{\it two-way finite automata with quantum and classical states} (2QCFA)  in 2002. As an intermediate
model between 1QFA and 2QFA, 2QCFA are still more powerful than
their classical counterparts.  A 2QCFA is essentially a
classical 2DFA augmented with a quantum component of constant
size, where the dimension of the associated Hilbert space does not
depend on the length of the input. In spite of the
existing restriction, 2QCFA are more powerful than 2PFA.  Indeed, 2QCFA can
recognize all regular languages with certainty, and particularly, Ambainis and Watrous  \cite{AJ} proved
that this model can also recognize $L_{eq}$ with one-sided error in polynomial expected time and can recognize  palindromes $L_{pal}=\{\omega\in\{a,b\}^{*}\mid\omega=\omega^{R}\}$ with one-sided error in exponential expected time. Note that   no 2PFA can recognize $L_{pal}$ with bounded error.

Several open problems were proposed by  Ambainis and Watrous \cite{AJ}, including the problem whether $L_{middle}=\{xay\mid x,y\in\Sigma^{*},a\in\Sigma, |x|=|y|\}$ can be recognized by 2QCFA or not. In this paper, we does not aim to answer the above question, but we consider a  similar language  $L_{m}=\{xcy\mid  \Sigma=\{a, b, c\}, x,y\in\{a,b\}^{*},c\in\Sigma, |x|=|y| \}$. We prove that  $L_{m}$ can be recognized by 2QCFA with one-sided error in polynomial expected time. Meanwhile we show that $L_{m}$ can also be recognized by 2PFA with bounded error, but in exponential expected time.  Thus $L_{m}$ is another witness of the fact that 2QCFA are more powerful than their classical counterparts.

The remainder of this paper is organized as follows. Some computing models and related definitions are introduced in Section 2. In section 3 we describe a 2QCFA for recognizing $L_m$ with one-sided error in polynomial expected time. In section 4 we show $L_m$ can be recognized by 2PFA with bounded error in exponential expected time. Finally, some concluding remarks are made in Section 5.

\section{Definitions}
We recall the definitions of 2PFA and 2QCFA in this section.

\subsection{Definition of two-way probabilistic finite automata}
The notation of 2PFA was  introduced by Kuklin \cite{YI}, and then studied by  Freivalds \cite{RF} and Dwork etc \cite{CL,CL2,AA}.

A 2PFA $\mathcal{M}$ is defined by a 6-tuple
\begin{equation}
\mathcal{M}=(S,\Sigma,\delta,s_{0},S_{acc},S_{rej})
\end{equation}
where,
\begin{itemize}

\item $S$ is a finite set of classical states;

\item $\Sigma$ is a finite set of input symbols;  the tape symbol
set $\Gamma=\Sigma\cup\{\ |\hspace{-1.5mm}c,\$\}$, where $\ |\hspace{-1.5mm}c\notin \Sigma $ is called the left end-marker and $\$\notin \Sigma$ is called the right end-marker;

\item $s_{0}\in S$ is the initial state of the machine;

\item $S_{acc}\subset S$ and $S_{rej}\subset S$ are the sets of
accepting states and rejecting states, respectively.

\item $\delta$ is the transition function:

\begin{equation}
(S\setminus(S_{acc} \cup S_{rej}))\times \Gamma \times S \times \{-1, 0, 1\} \rightarrow \{0, 1/2, 1\}
\end{equation}
 Essentially, for each state $s\in S$ and each $\sigma\in\Sigma\cup \{\ |\hspace{-1.5mm}c,\$\}$, $\delta(s,\sigma)$ is a coin-tossing distribution\footnote{A coin-tossing distribution on finite set $Q$ is a mapping $\phi$ from $Q$ to \{0, 1/2, 1\} such that $\sum_{q\in Q}\phi(q)=1$, which means choosing $q$ with probability $\phi(q)$.} on $S \times \{-1, 0, 1\}$, where $d=-1$ means that the tape head moves one square left, $d=0$ means that the tape head keeps stationary, and $d=1$ means that the tape head moves one square right. We assume that $\delta$ is well defined so that when the tape head is positioned on the left end-marker $\ |\hspace{-1.5mm}c$ (right end-marker $\$$), the tape head will not move left (right) in next step.

\end{itemize}

The computation of a 2PFA $\mathcal{M}$ on input $\omega\in\Sigma^*$ begins with the initial state $s_0$ and with the word $\ |\hspace{-1.5mm}c\omega\$$ written on the tape where the tape head is positioned on the left end-marker $\ |\hspace{-1.5mm}c$. The computation is then governed (probabilistically) by the transition functions $\delta$ until $\mathcal{M}$ either accepts $\omega$ by entering an accepting state $s_{a}\in S_{acc}$ or rejects $\omega$ by entering a rejecting state $s_{r}\in S_{rej}$. $\mathcal{M}$ halts when it enters an accepting state or a rejecting state. It should be pointed out that the computation could be infinite if neither an accepting state nor a rejecting state is entered. Let $L\subset \Sigma^*$ and $0\leq\epsilon<1/2$. Then a 2PFA $\mathcal{M}$ recognizes $L$ with bounded error if
\begin{itemize}
\item[1.] $\forall\omega\in L$, $Pr[\mathcal{M}\  accepts\  \omega]\geq 1-\epsilon$, and
\item[2.] $\forall\omega\notin L$, $Pr[\mathcal{M}\  rejects\  \omega]\geq 1-\epsilon$.
\end{itemize}

\subsection{Definition of two-way finite automata with quantum and classical states}
2QCFA were introduced  by Ambainis and Watrous \cite{AJ} in 2002, and then studied by in \cite{Qiu2,AY2,Zheng}.

Informally, we describe a 2QCFA as a 2DFA which has access
to a constant size of quantum register, upon which it performs
quantum transformations and measurements. We would refer the readers to
\cite{Gru99,Nie} for a detailed overview of quantum
computing.

A 2QCFA is specified by a 9-tuple
\begin{equation}
\mathcal{M}=(Q,S,\Sigma,\Theta,\delta,q_{0},s_{0},S_{acc},S_{rej})
\end{equation}
where,
\begin{itemize}
\item $Q$ is a finite set of quantum states;

\item $S$ is a finite set of classical states;

\item $\Sigma$ is a finite set of input symbols;  the tape symbol
set $\Gamma=\Sigma\cup\{\ |\hspace{-1.5mm}c,\$\}$, where $\ |\hspace{-1.5mm}c\notin \Sigma $ is called the left end-marker and $\$\notin \Sigma$ is called the right end-marker;

\item $q_{0}\in Q$ is the initial quantum state;

\item $s_{0}\in S$ is the initial classical state;

\item $S_{acc}\subset S$ and $S_{rej}\subset S$ are the sets of
classical accepting states and rejecting states, respectively.

\item $\Theta$ is the transition function of quantum states:
\begin{equation}
S\setminus(S_{acc}\cup S_{rej})\times \Gamma \rightarrow
\mathcal{U}(\mathcal{H}(Q))\cup \mathcal{M}(\mathcal{H}(Q)),
\end{equation}
where
$\mathcal{U}(\mathcal{H}(Q))$ and $\mathcal{M}(\mathcal{H}(Q))$ respectively
denote the sets of
 unitary operators and projective measurements over $\mathcal{H}(Q)$,  and $\mathcal{H}(Q)$ represents the Hilbert space with
the corresponding base identified with set $Q$. Thus,
 $\Theta(s,\gamma)$ corresponds to either a
unitary transformation or a projective measurement.

\item $\delta$ is the transition function of classical states. If
$\Theta(s, \gamma)\in \mathcal{U}(\mathcal{H}(Q))$, then $\delta$
is
\begin{equation}
S\setminus(S_{acc}\cup S_{rej})\times \Gamma \rightarrow S\times
\{-1, 0, 1\},
\end{equation}

which is similar to the transition function defined for 2PFA except that the transition function here is deterministic.
If $\Theta(s, \gamma)\in \mathcal{M}(\mathcal{H}(Q))$ which is a projective measurement,  assume that the projective measurement with the set of possible eigenvalues $R=\{r_1, \cdots, r_n\}$ and
the projector set $\{P(r_i):i=1, \cdots, n\}$ where $P(r_i)$ denotes the projector onto the eigenspace corresponding to $r_i$,
 the measurement result set will be  $R=\{r_{1},r_{2},\dots$,
$r_{n}\}$. Then
 $\delta$ is
\begin{equation}
S\setminus(S_{acc}\cup S_{rej})\times \Gamma \times R\rightarrow
S\times \{-1, 0, 1\},
\end{equation}
 where $\delta(s,\gamma)(r_{i})=(s',d)$ means
that when the projective measurement result is $r_{i}$, the classical state  $s\in S$ scanning $\gamma\in \Gamma$
is changed to state $s'$, and the movement of the tape head is
decided by $d$.
\end{itemize}

Given an input $\omega$, a 2QCFA
$\mathcal{M}=(Q,S,\Sigma,\Theta,\delta,q_{0},s_{0},S_{acc},S_{rej})$
proceeds as follows:

At the beginning, the tape head is positioned on\ $|\hspace{-1.5mm}c$,
the quantum initial state is $|q_{0}\rangle$, the classical
initial state is $s_{0}$, and $|q_{0}\rangle$ will be changed
according to $\Theta(s_{0},\ |\hspace{-1.5mm}c)$.
\begin{itemize}

\item [a.] If $\Theta(s_{0},\ |\hspace{-1.5mm}c)=U\in
\mathcal{U}(\mathcal{H}(Q))$, then the quantum state evolves as
$|q_{0}\rangle\rightarrow U|q_{0}\rangle$, and meanwhile, the
classical state $s_{0}$ is changed to $s'$ according to
$\delta(s_{0},\ |\hspace{-1.5mm}c_{1})=(s',d)$. The movement of the tape head is
decided by $d$.

\item [b.] If $\Theta(s_{0},\ |\hspace{-2.0mm}C_{1})=M\in
\mathcal{M}(\mathcal{H}(Q))$, then the measurement $M$ is performed on $|q_{0}\rangle$.
Let $M=\{P_1,\cdots, P_m\}$ with result set $R=\{r_i\}_{i=1}^m$.  After the measurement $M$ has been performed, we get a result $r_i\in R$ with probability $p_i=\langle q_0|P_i|q_0\rangle$, and the quantum state $|q_0\rangle$ changes to ${P_i|q_0\rangle}/{\sqrt{p_i}}$.
Meanwhile, the classical state changes according to
$\delta(s_{0},\
|\hspace{-1.5mm}c)(r_{i})=(s_{i},d)$. If $s_{i}\in
S_{acc}$ $(S_{rej})$, $\mathcal{M}$ accepts (rejects) $\omega$ and halts; otherwise, the tape head of $\mathcal{M}$ moves according to the direction $d$, and continues to read the next symbol.
\end{itemize}
A computation is assumed to halt if and only if an accepting state
or a rejecting classical state is entered.

Let $L\subset \Sigma^*$ and $0\leq\epsilon<1/2$. A 2QCFA $\mathcal{M}$ recognizes $L$ with one-sided error if

\begin{itemize}
\item[1.] $\forall\omega\in L$, $Pr[\mathcal{M}\  accepts\  \omega]=1$, and
\item[2.] $\forall\omega\notin L$, $Pr[\mathcal{M}\  rejects\  \omega]\geq 1-\epsilon$.
\end{itemize}

\section {A 2QCFA recognizing the language   $L_{m}$ }
We prove that $L_{m}=\{xcy\mid  \Sigma=\{a, b, c\}, x,y\in\{a,b\}^{*},c\in\Sigma, |x|=|y|\}$ can be recognized by 2QCFA with one-sided error in polynomial expected time in this section.

\begin{theorem}\label{th1}
For any $\epsilon>0$, there is a 2QCFA $\mathcal{M}$ that accepts
any $\omega \in L_{m}=\{xcy\mid  \Sigma=\{a, b, c\}, x,y\in\{a,b\}^{*},c\in\Sigma, |x|=|y| \}$
with certainty, rejects any $\omega \notin
L_{m}$ with
probability at least $1-\epsilon$ and halts in polynomial expected time.
\end{theorem}
\begin{proof}
In order to prove Theorem \ref{th1}, we consider two matrices $U_{a}$ and $U_{c}$ defined as follows:
\begin{equation}\label{e7}
 U_{a}=\left(
  \begin{array}{cccc}
    \cos\alpha & -\sin\alpha & 0 & 0 \\
    \sin\alpha & \cos\alpha & 0 & 0 \\
    0& 0 & \cos\alpha & \sin\alpha\\
    0& 0 & -\sin\alpha & \cos\alpha  \\
  \end{array}
\right),
U_{c}=\left(
  \begin{array}{cccc}
    0 & 0 & 1 & 0 \\
    0 & 0 & 0 & 1 \\
    1& 0 & 0 & 0\\
    0& 1 & 0& 0  \\
  \end{array}
\right)
\end{equation}
where let $\alpha=\sqrt 2 \pi $.
Obviously, $U_{a}$ and $U_{c}$ are unitary, and  it is easy to get that
\begin{equation}\label{e8}
(U_{a})^k=\left(
  \begin{array}{cccc}
    \cos k\alpha & -\sin k\alpha & 0 & 0 \\
    \sin k\alpha & \cos k\alpha & 0 & 0 \\
    0& 0 & \cos k \alpha & \sin k\alpha\\
    0& 0 & -\sin k\alpha & \cos k \alpha  \\
  \end{array}
\right),
\end{equation}

and
\begin{equation}\label{e9}
(U_{a})^{l}U_{c}(U_{a})^{k}=\left(
  \begin{array}{cccc}
    0 & 0& \cos(k-l)\alpha  & \sin(k-l)\alpha \\
    0 & 0& -\sin(k-l)\alpha & \cos(k-l)\alpha \\
    \cos(k-l)\alpha& -\sin(k-l)\alpha & 0 & 0\\
    \sin(k-l)\alpha& \cos(k-l)\alpha  & 0 & 0  \\
  \end{array}
\right).
\end{equation}

 We now describe a 2QCFA $\mathcal{M}$ with 4 quantum states $\{|q_0\rangle, |q_1\rangle, |q_2\rangle,|q_3\rangle\}$, of which  $|q_0\rangle$ is the initial state.  $\mathcal{M}$  has three unitary operators: $U_{a}$, $U_{b}$ and $U_{c}$ where $U_{a}$ and $U_{c}$ are given in Eq. (\ref{e7}) and $U_{b}=U_{a}$. They can also be described as follows: \\

  \begin{tabular}{|l|l|l|}
  \hline
  $U_{a}|q_0\rangle=\cos \alpha |q_0\rangle+\sin \alpha|q_1\rangle$ & $U_{b}|q_0\rangle=\cos \alpha |q_0\rangle+\sin \alpha|q_1\rangle$&$U_{c}|q_0\rangle=|q_2\rangle$ \\
  $U_{a}|q_1\rangle=-\sin \alpha |q_0\rangle+\cos \alpha|q_1\rangle$ & $U_{b}|q_1\rangle=-\sin \alpha |q_0\rangle+\cos \alpha|q_1\rangle$&$U_{c}|q_1\rangle=|q_3\rangle$ \\
  $U_{a}|q_2\rangle=\cos \alpha |q_2\rangle-\sin \alpha|q_3\rangle$ & $U_{b}|q_2\rangle=\cos \alpha |q_2\rangle-\sin \alpha|q_3\rangle$&$U_{c}|q_2\rangle=|q_0\rangle$ \\
  $U_{a}|q_3\rangle=\sin \alpha |q_2\rangle+\cos \alpha|q_3\rangle$ &$U_{b}|q_3\rangle=\sin \alpha |q_2\rangle+\cos \alpha|q_3\rangle$ &$U_{c}|q_3\rangle=|q_1\rangle$ \\

  \hline
\end{tabular}\\

The automaton  $\mathcal{M}$ proceeds as follows:\\
\begin{tabular}{|l|}
  \hline
Check whether the input is of the form $xcy$ ($x,y\in\Sigma^*$). If not, reject.\\
Otherwise, repeat the following ad infinitum:\\
\ \ 1. Move the tape head to the first input symbol and set the quantum state to $|q_0\rangle$.\\
\ \ 2. While the currently scanned symbol is not $\$$, do the following:\\
\ \ \ \ (2-1). If the currently scanned symbol is $a$ or $b$, perform $U_a$ on the quantum state.\\
\ \ \ \ (2-2). If the currently scanned symbol is $c$, perform $U_c$ on the quantum state.\\
\ \ \ \ (2-3). Move the tape head one square to the right.\\
\ \ 3. Measure the quantum state. If the result is not $|q_2\rangle$, reject.\\
\ \ 4. Repeat the following subroutine two times:\\
\ \ \ \ (4-1).Move the tape head to the first input symbol.\\
\ \ \ \ (4-2).Move the tape head one square to the right.\\
\ \ \ \ (4-3).While the currently scanned symbol is not $\ |\hspace{-1.5mm}c$ or $\$$, do the following:\\
\ \ \ \ \ \ \ \ Simulate a coin flip. If the result is ``head", move right. Otherwise, move left.\\
\ \ 5. If both times the process ends at the right end-marker $\$$, do:\\
\ \ \ \ \ \ Simulate $k$ coin-flips. If all results are ``heads", accept.\\
  \hline
\end{tabular}\\

\begin{lemma}
If the input $\omega=xcy$ satisfies $|x|=n,$ $|y|=m$ and $n=m$, then the quantum state of $\mathcal{M}$ will evolve into $|q_2\rangle$ after loop \textbf{2} with certainty.
\end{lemma}
\begin{proof}
According to Eq. (\ref{e9}), the quantum state after loop \textbf{2} can be described as follows:
\begin{equation}
|q\rangle=(U_{a})^mU_c(U_a)^n|q_0\rangle
\end{equation}
\begin{equation}
=\left(
  \begin{array}{cccc}
    0 & 0& \cos(n-m)\alpha  & \sin(n-m)\alpha \\
    0 & 0& -\sin(n-m)\alpha & \cos(n-m)\alpha \\
    \cos(n-m)\alpha& -\sin(n-m)\alpha & 0 & 0\\
    \sin(n-m)\alpha& \cos(n-m)\alpha  & 0 & 0  \\
  \end{array}
\right)|q_0\rangle.
\end{equation}
Because $n=m$, we get
\begin{equation}
|q\rangle=\left(
  \begin{array}{cccc}
    0 & 0& 1 & 0 \\
    0 & 0& 0 & 1 \\
    1 & 0& 0 & 0\\
    0 & 1& 0 & 0  \\
  \end{array}
\right)|q_0\rangle=\left(
  \begin{array}{cccc}
    0 & 0& 1 & 0 \\
    0 & 0& 0 & 1 \\
    1 & 0& 0 & 0\\
    0 & 1& 0 & 0  \\
  \end{array}
\right)
\left(
 \begin{array}{c}
    1 \\
    0 \\
    0 \\
    0 \\
  \end{array}
\right)
=\left(
 \begin{array}{c}
    0 \\
    0 \\
    1 \\
    0 \\
  \end{array}
\right)=|q_2\rangle.
\end{equation}
So the lemma has been proved.
\end{proof}

\begin{lemma}\label{lm2}
If the input $\omega=xcy$ satisfies $|x|=n$, $|y|=m$ and  $n\neq m$, then $\mathcal{M}$ rejects after step \textbf{3}  with probability at least $1/(2(n-m)^2+1)$.
\end{lemma}

\begin{proof}
Starting with state $|q_0\rangle$, $\mathcal{M}$ changes its quantum state to  $(U_a)^mU_c(U_a)^n|q_0\rangle$ after loop \textbf{2} .
According to the analysis given above and Eq. (\ref{e8}-\ref{e9}), we get the quantum state
\begin{equation}
|q\rangle=(U_a)^mU_c(U_a)^n|q_0\rangle=\cos((n-m)\alpha)|q_2\rangle+\sin((n-m)\alpha)|q_3\rangle
\end{equation}
\begin{equation}
=\cos(\sqrt{2}(n-m)\pi)|q_2\rangle+\sin(\sqrt{2}(n-m)\pi)|q_3\rangle.
\end{equation}

 The probability of observing  $|q_3\rangle$ is
$\sin^2(\sqrt{2}(n-m)\pi)$ in step \textbf{3}. Without loss of generality, we assume that
$n-m>0$. Let $l$ be the closest integer to $\sqrt{2}(n-m)$. If $\sqrt{2}(n-m)>l$, then
$2(n-m)^2>l^2$. So we get $2(n-m)^2-1\geq l^2$ and $l\leq
\sqrt{2(n-m)^2-1}$. We have
\begin{equation}
\sqrt{2}(n-m)-l\geq \sqrt{2}(n-m)-\sqrt{2(n-m)^2-1}
\end{equation}

\begin{equation}
=\frac{(\sqrt{2}(n-m)-\sqrt{2(n-m)^2-1})(\sqrt{2}(n-m)+\sqrt{2(n-m)^2-1})}{\sqrt{2}(n-m)+\sqrt{2(n-m)^2-1}}
\end{equation}
\begin{equation}
=\frac{1}{\sqrt{2}(n-m)+\sqrt{2(n-m)^2-1}}>\frac{1}{2\sqrt{2}(n-m)}.
\end{equation}

Because $l$ is the closest integer to $\sqrt{2}(n-m)$, we have
$0<\sqrt{2}(n-m)-l<1/2$. Let $f(x)=sin(x\pi)-2x$. We have
 $f''(x)=-\pi^2\sin(x\pi)\leq 0$ when $x\in
[0,1/2]$. That is to say, $f(x)$ is concave in $[0,1/2]$, and we
have $f(0)=f(1/2)=0$. So for any $x\in[0,1/2]$, it holds that
$f(x)\geq 0$, that is, $\sin(x\pi)\geq 2x$. Therefore, we have
\begin{equation}
\sin^2(\sqrt{2}(n-m)\pi)=\sin^2((\sqrt{2}(n-m)-l)\pi)
\end{equation}
\begin{equation}
\geq (2(\sqrt{2}(n-m)-l))^2 =4(\sqrt{2}(n-m)-l)^2
\end{equation}
\begin{equation}
> 4(\frac{1}{2\sqrt{2}(n-m)})^2=\frac{1}{2(n-m)^2}>\frac{1}{2(n-m)^2+1}.
\end{equation}

If $\sqrt{2}(n-m)<l$, then
$2(n-m)^2<l^2$. So we get $2(n-m)^2+1\leq l^2$ and $l\geq
\sqrt{2(n-m)^2+1}$. We have
\begin{equation}
\sqrt{2}(n-m)-l\leq \sqrt{2}(n-m)-\sqrt{2(n-m)^2+1}
\end{equation}
\begin{equation}
=\frac{(\sqrt{2}(n-m)-\sqrt{2(n-m)^2+1})(\sqrt{2}(n-m)+\sqrt{2(n-m)^2+1})}{\sqrt{2}(n-m)+\sqrt{2(n-m)^2+1}}
\end{equation}
\begin{equation}
=\frac{-1}{\sqrt{2}(n-m)+\sqrt{2(n-m)^2+1}}<\frac{-1}{2\sqrt{2(n-m)^2+1}}.
\end{equation}

It follows that
\begin{equation}
l-\sqrt{2}(n-m)>\frac{1}{2\sqrt{2(n-m)^2+1}}.
\end{equation}
Because $l$ is the closest integer to $\sqrt{2}(n-m)$, we have
$0<l-\sqrt{2}(n-m)<1/2$. Therefore, we have
\begin{equation}
\sin^2(\sqrt{2}(n-m)\pi)=\sin^2((\sqrt{2}(n-m)-l)\pi)
\end{equation}
\begin{equation}
=\sin^2((l-\sqrt{2}(n-m))\pi)\geq (2(l-\sqrt{2}(n-m)))^2
\end{equation}
\begin{equation}
=4(l-\sqrt{2}(n-m))^2> 4(\frac{1}{2\sqrt{2(n-m)^2+1}})^2=\frac{1}{2(n-m)^2+1}.
\end{equation}

So the lemma has been proved.
\end{proof}

Simulation of a coin flip in loops \textbf{4} and \textbf{5} is a key component in the above algorithm. We will show that coin-flips can be simulated by 2QCFA.
\begin{lemma}
A coin flip can be simulated by 2QCFA $\mathcal{M}$ with a unitary operation and a measurement.
\end{lemma}
\begin{proof}
 A projective measurement $M=\{P_0,P_1\}$  is defined by
\begin{equation}
 P_0=|p_0\rangle\langle p_0|, P_1=|p_1\rangle\langle p_1|.
\end {equation}
The results 0 and 1 represent the
``tail'' and ``head''  of a coin flip, respectively. A unitary operator $U$ is given by
\begin{equation}
  U=\left(%
\begin{array}{cc}
  \frac{1}{\sqrt{2}} &  \frac{1}{\sqrt{2}} \\
   \frac{1}{\sqrt{2}} &  -\frac{1}{\sqrt{2}} \\
\end{array}%
 \right).
\end {equation}
The unitary operator $U$ changes the state as
 follows:
\begin{equation}
|p_0\rangle\rightarrow|\psi\rangle=\frac{1}{\sqrt{2}}(|p_0\rangle+|p_1\rangle),\ \  |p_1\rangle\rightarrow|\phi\rangle=\frac{1}{\sqrt{2}}(|p_0\rangle-|p_1\rangle).
\end{equation}

Suppose now that the machine starts with the state $|p_0\rangle$, changes its state by $U$, and then measures the state with $M$. Then  we will get the result  0 or 1 with probability
 $\frac{1}{2}$. This is similar to a coin flip
 process.
\end{proof}

\begin{lemma}
Every execution of loops \textbf{4} and \textbf{5}  leads to acceptance with  probability $1/2^k(n+m+2)^2$.
\end{lemma}
\begin{proof}
Loop \textbf{4} is two times of random walk starting at location 1 and
ending at location 0 (the left end-marker\ \
$|\hspace{-1.5mm}c$) or at location $n+m+2$ (the right
end-marker $\$$). It can be known from probability theory that the
probability of reaching the location $n+m+2$ is $1/(n+m+2)$
(see Chapter14.2 in \cite{WF}). Repeating it twice and flipping
$k$ coins, we  get the probability $1/2^k(n+m+2)^2$.
\end{proof}

Let $k=1+\lceil\log_2 1/\varepsilon\rceil$, then $\varepsilon\geq 1/2^{k-1}$. If
$\omega=xcy$ satisfies $|x|=|y|=n$, loop $\mathbf{2}$ always changes
$|q_0\rangle$ to $|q_2\rangle$, and $\mathcal{M}$ never rejects after the measurement in step $\mathbf{3}$.
After loops $\mathbf{4}$ and $\mathbf{5}$, the probability of $\mathcal{M}$
accepting $\omega$ is $1/2^k(2n+2)^2$. Repeating loops $\mathbf{4}$ and $\mathbf{5}$
for $cn^2$ times, the accepting probability is
\begin{equation}
Pr[\mathcal{M}\  accepts\  \omega] =1-(1-\frac{1}{2^k(2n+2)^2})^{cn^2},
\end{equation}
 and this can be made
arbitrarily close to 1 by selecting constant $c$ appropriately.

Otherwise, if $\omega=xcy$ satisfies $|x|=n$, $|y|=m$ and $n\neq m$, $\mathcal{M}$ rejects after
loop $\mathbf{2}$ and step $\mathbf{3}$  with probability
\begin{equation}
P_{r} >\frac{1}{2(n-m)^2+1}
\end{equation}
 according to Lemma \ref{lm2}. $\mathcal{M}$ accepts after loops $\mathbf{4}$ and $\mathbf{5}$ with probability
\begin{equation}
P_{a}=1/2^k(n+m+2)^2\leq \varepsilon/2(n+m+2)^2.
\end{equation}
If we repeat the whole algorithm indefinitely, the probability of $\mathcal{M}$ rejecting input $\omega$ is
\begin{equation}
Pr[\mathcal{M}\  rejects\  \omega] =\sum_{i\geq 0}(1-P_a)^i(1-P_r)^iP_r
\end{equation}
\begin{equation}
=\frac{P_r}{P_a+P_r-P_aP_r}>\frac{P_r}{P_a+P_r}
\end{equation}
\begin{equation}
>\frac{1/(2(n-m)^2+1)}{\varepsilon/2(n+m+2)^2+1/(2(n-m)^2+1)}
\end{equation}
\begin{equation}
>\frac{1/2}{1/2+\varepsilon/2}=\frac{1}{1+\varepsilon}>1-\varepsilon.
\end{equation}

If we assume that the input $\omega=xcy$ where $|x|=n$ and $|y|=m$, then loop $\mathbf{1}$ takes
$\mathbf{O}(n+m)$ time at worst cases, loop $\mathbf{2}$ and step $\mathbf{3}$ take
$\mathbf{O}(n+m)$ time exactly, and loops $\mathbf{4}$ and $\mathbf{5}$ take
$\mathbf{O}((n+m)^2)$ time. The expected number of repeating the algorithm
 is $\mathbf{O}((n+m)^2)$. Hence, the expected running
time of $\mathcal{M}$ is $\mathbf{O}((n+m)^4)$.
\end{proof}

\section {A 2PFA recognizing the language $L_m$ }

$L_m$ looks like the language $L_{middle}=\{xay\mid x,y\in\Sigma^{*},a\in\Sigma, |x|=|y|\}$ which has been proved  to be not recognizable by 2PFA in Dwork and Stockmeyer's paper\cite{CL2}. However, we will prove that $L_m$ can be recognized by 2PFA in this section, but the expected time needed is exponential.

\begin{theorem}
For any $\epsilon>0$, there is a 2PFA $\mathcal{M}$ that accepts
any $\omega \in L_{m}=\{xcy\mid  \Sigma=\{a, b, c\}, x,y\in\{a,b\}^{*},c\in\Sigma, |x|=|y| \}$
at least $1-\epsilon$, rejects any $\omega \notin
L_{m}$ with
probability at least $1-\epsilon$.
\end{theorem}
\begin{proof}
 We assume that the input $\omega\in \{a, b, c\}^*$ has the form $\omega=xcy$ where $|\omega|=l$, $|x|=n$, and $|y|=m$. Let $k$ be a positive integer. The algorithm for a 2PFA to recognize $L_m$ is described as follows:\\
\begin{tabular}{|l|}
  \hline
Checks whether the length of input $l=|\omega|$ is odd. If not, rejects.\\
Otherwise, repeat the following ad infinitum:\\
1. Move the tape head to symbol $c$. If there is not symbol $c$ in $\omega$, rejects.\\
2. Simulate a coin flip, and do the following:\\
(2-1). If the outcome is ``head", simulate $k(2n+2)$ coin-flips, \\
and move the tape head left to keep count.\\
(2-2). Otherwise, simulate $k(2m+2)$  coin-flips, \\
\ \ \ \ \ \ \ \ \ and move its tape head right to keep count.\\
(2-3). If all $k(2n+2)$ or $k(2m+2)$ flips have outcome ``heads" in either case,\\
\ \ \ \ \ \ \ \ \   reject. \\
3. Simulate $kl$ coin-flips using the input $\omega$ to keep count. \\
\ \ \ \ \ \ \ \ \  If all the outcomes are ``heads", accept.\\
  \hline
\end{tabular}\\

We argue that this algorithm is a 2PFA for $L_m$. Consider first the case that $\omega\in L_m$.
At each iteration, n=m. The probability of $\mathcal{M}$ rejecting $\omega$ in an iteration is
\begin{equation}
P_r=2^{-k(2n+2)},
\end{equation}
and the probability of $\mathcal{M}$ accepting $\omega$ in an iteration  is
\begin{equation}
P_a=(1-2^{-k(2n+2)})\times 2^{-kl}\geq 2^{-1}2^{-kl}=2^{-k(2n+1)-1}.
\end{equation}
Repeating the iteration indefinitely, causes $\mathcal{M}$ to eventually accept $\omega$ with probability
\begin{equation}
Pr[\mathcal{M}\  accepts\  \omega] =\sum_{i\geq 0}(1-P_a)^i(1-P_r)^{i+1}P_a
\end{equation}
\begin{equation}
=\frac{P_a-P_aP_r}{P_a+P_r-P_aP_r}=\frac{P_a(1-P_r)}{P_a(1-P_r)+P_r}
\end{equation}
\begin{equation}
= \frac{P_a}{P_a+P_r/(1-P_r)}\geq \frac{P_a}{P_a+2P_r}
\end{equation}
\begin{equation}
\geq \frac{2^{-k(2n+1)-1}}{2^{-k(2n+1)-1}+2\times2^{-k(2n+2)}}= \frac{1}{1+2^{-(k+2)}}.
\end{equation}

Therefore, the probability that $\mathcal{M}$ accepts before it rejects approaches 1 as $k$ increases.

Consider now the other case that  $\omega\notin L_m$. At each iteration, we have $n\neq m$. Therefore, either $2n+2\leq l-1$ or $2m+2\leq l-1$, and whichever case holds, $\mathcal{M}$ will choose the case with probability $1/2$. The probability of $\mathcal{M}$ rejecting $\omega$ in an iteration is
\begin{equation}
P_r=2^{-1}2^{-k(2n+2)}+2^{-1}2^{-k(2m+2)}\geq 2^{-1}2^{-k(l-1)},
\end{equation}
and the probability of $\mathcal{M}$ accepting $\omega$ in an iteration  is
\begin{equation}
P_a=(1-P_r)\times 2^{-kl}\leq 2^{-kl}.
\end{equation}
Repeating the iteration indefinitely, causes $\mathcal{M}$ to eventually reject $\omega$ with probability
\begin{equation}
Pr[\mathcal{M}\  rejects\  \omega] =\sum_{i\geq 0}(1-P_a)^i(1-P_r)^iP_r
\end{equation}
\begin{equation}
=\frac{P_r}{P_a+P_r-P_aP_r}\geq \frac{P_r}{P_a+P_r}
\end{equation}
\begin{equation}
\geq \frac{2^{-1}2^{-k(l-1)}}{2^{-kl}+2^{-1}2^{-k(l-1)}}=\frac{1}{2^{-(k-1)}+1}.
\end{equation}
Therefore, the probability that $\mathcal{M}$ rejects before it accepts approaches 1 as $k$ increases.

If we assume that the length of the input $|\omega|=l$, then each iteration takes $\mathbf{O}(l)$ time.
The expected number of repeating the algorithm
 is $2^{\mathbf{O}(l)}$. Hence, the expected running
time of $\mathcal{M}$ is $\mathbf{O}(l)2^{\mathbf{O}(l)}$, which is exponential in $l$.
\end{proof}

\begin{remark} It is easy to  show that $L_m$ is non-regular by using the pumping lemma of  regular languages.\end{remark}

In the above theorem, we showed that $L_m$ can be recognized by a 2PFA in exponential expected time $\mathbf{O}(l)2^{\mathbf{O}(l)}$ where $l$ is the length of input. Note that, any 2PFA needs exponential expected time to recognize it, since it is a {\it non-regular language} \cite{AA}. However, we have shown that $L_m$ can be recognized by a 2QCFA in polynomial expected time.  Hence, 2QCFA show superiority over 2PFA  in recognizing the language $L_m$.

\section*
{ACKNOWLEDGMENT}

The authors are thankful to the anonymous referees and editor
for their comments and suggestions that greatly help to improve
the quality of the manuscript.

\end{document}